\documentclass{llncs}

\usepackage[utf8]{inputenc}
\usepackage[T1]{fontenc}
\usepackage[english]{babel}
\usepackage{lmodern}
\usepackage{amsmath,amstext,amssymb}
\usepackage{graphicx}
\usepackage{listings}
\usepackage{url} %

\usepackage{color}   %
\usepackage{parskip} %
\usepackage{ntheorem}  %
\usepackage{comment}  %

\usepackage{tikz}
\usetikzlibrary{arrows,backgrounds,positioning,fit,scopes,chains,matrix}

\newcommand{\arrow}[1]{\overrightarrow{#1}}

\newcommand{\set}[1]{\operatorname{\textsc{#1}}}

\newcommand{\true}[0]{\textsc{True}}
\newcommand{\false}[0]{\textsc{False}}

\newcommand{\code}[1]{\operatorname{\texttt{#1}}}
\newcommand{\decode}[0]{\operatorname{\sf decode}}
\newcommand{\encode}[0]{\operatorname{\sf encode}}
\newcommand{\abstractValue}[0]{\operatorname{\sf abstract-value}}
\newcommand{\concreteValue}[0]{\operatorname{\sf concrete-value}}
\newcommand{\compile}[0]{\operatorname{\sf compile}}
\newcommand{\flags}[0]{\operatorname{\sf flags}}

\newcommand{\argument}[0]{\operatorname{\sf argument}}
\newcommand{\arguments}[0]{\operatorname{\sf arguments}}
\newcommand{\merge}[0]{\operatorname{\sf merge}}
\newcommand{\absEnv}[0]{E_\mathbb{A}}

\newcommand{\defness}[0]{\operatorname{\sf def}}
\newcommand{\undefined}[0]{\operatorname{\sf undefined}}
\newcommand{\numeric}[0]{\operatorname{\sf numeric}}

\tikzset {
    every join/.style={rounded corners}
  , tip/.style={->,>=latex}
  , tipDouble/.style={->,>=latex,double}
  , tipThick/.style={->,>=latex,thick}
  , tipThin/.style={->,>=latex,thin}
  , labelabove/.style={font=\small,above}
  , labelbelow/.style={font=\small,below}
  , labelright/.style={font=\small,right}
  , labelleft/.style={font=\small,left}
  , data/.style={rectangle,rounded corners=1pt,draw=black!50,fill=black!10,thick}
  , library/.style={rectangle,rounded corners=1pt,draw=blue!50,fill=blue!20,thick}
  }

\title{Propositional Encoding 
  of Constraints over Tree-Shaped Data}
\author{Alexander Bau \and Johannes Waldmann \\ 
  \institute{HTWK Leipzig (F-IMN)\\P.O.B. 301166\\04251 Leipzig, Germany}
}
\begin{document}

\maketitle

\pagestyle{plain} 

\newcommand{\einput}[1]{\marginpar{file:\\ \texttt{#1}}\input{#1}}

\begin{abstract}
We present a functional programming language
for specifying constraints over tree-shaped data.
The language allows for Haskell-like algebraic data types
and pattern matching.
Our constraint compiler CO4 translates these programs
into satisfiability problems in propositional logic.
We present an application from the area of
automated analysis of (non-)termination of rewrite systems.
\end{abstract}

\section{Motivation}

The paper presents a high-level declarative language
CO4 for describing constraint systems.
The language includes user-defined
algebraic data types and recursive functions
defined by pattern matching,
as well as higher-order and polymorphic types.
This language comes with a compiler that transforms
a high-level constraint system into a satisfiability problem
in propositional logic. This is motivated by the following.

Recent years have seen a tremendous development 
of constraint solvers for propositional satisfiability (SAT solvers).
Based on the 
Davis-Putname-Logemann-Loveland %
algorithm and extended with 
conflict-driven clause learning, %
SAT solvers like Minisat \cite{DBLP:conf/sat/EenS03} are able to find satisfying
assignments for conjunctive normal forms with $10^6$ and more clauses
in a lot of cases. SAT solvers are used in industrial-grade 
verification of hardware and software.

With the availability of powerful SAT solvers,
\emph{propositional encoding} is a promising method
to solve constraint systems that originate in different domains.
In particular, this  approach had been used for automatically 
analysing (non-)termination of rewriting 
\cite{DBLP:conf/ieaaie/KuriharaK04,DBLP:conf/sofsem/ZanklSHM10,DBLP:journals/jar/CodishGST12}
successfully, as can be seen from the results of International
Termination Competitions (most of the participants use
propositional encodings).

So far, these encodings are written manually:
the programmer has to construct explicitly a formula in propositional logic
that encodes the desired properties.
This has the advantage that the formula can be optimized in clever ways,
but also the drawback that correctness of the formula is not evident,
so the process is error-prone. 

This is especially so if the data domain
for the constraint system is remote from the ``sequence of bits'' domain
that naturally fits propositional logic.
In typical applications, data is tree-structured 
(e.g., terms, and lists of terms)
and one wants to write constraints on such data in a direct way.

Our language is similar to Haskell \cite{haskell}
in the following sense: CO4 syntactically is a subset of Haskell
(including data declarations, case expressions, higher order functions,
polymorphism, but no type classes), and semantically CO4 is evaluated strictly.

The advantages of re-using a high level declarative language
for expressing constraint systems are:
the programmer can rely on established syntax and semantics,
does not have to learn a new language, can re-use his experience
and intuition, and can re-use actual code.

For instance, the (Haskell) function that describes
the application of a rewrite rule at some position in some
term can be directly used in a constraint system
that describes a rewrite sequence with a certain property.
We treat this application in detail in Section~\ref{sec:loop},
but need some preparation first.

A constraint programming language needs some way of parametrizing
the constraint system to data that is not available
when writing the program. For instance, a constraint program
for finding looping derivations for a rewrite system $R$,
will not contain a fixed system $R$, but will get $R$ as run-time input.

To accomodate for such applications, 
CO4 programs are handled and executed in two stages: 
The input program defines a function of type

\begin{equation*}
  f : K \times U \to \{\false,\true\}
\end{equation*}

where $K$ is some parameter domain (e.g., rewrite systems)
and $U$ is the domain of the unknown object (e.g., derivations).
In the first processing stage (at \emph{compile-time}), 
the program for $f$ is translated into a program

\begin{equation*}
  g : K \to (F, \Sigma \to U)
\end{equation*}

with $F$ being the set of formulas of propositional logic,
and $\Sigma$ being the set of assignments from variables of $F$ to truth values.

In the second stage (at \emph{run-time}),
a parameter value $p \in K$ is given,
and $g\ p$ is evaluated to produce a pair $(v,d) \in (F, \Sigma \to U)$.
An external SAT solver then tries to determine a satisfying assignment 
$\sigma \in \Sigma$ of $v$.
On success, $d(\sigma)$ is evaluated to a \emph{solution} value $s \in U$.
Proper compilation ensures that $f\ p\ s = \true$.

A formal specification of compilation is given in Section~\ref{sec:semantics},
and a concrete realization of compilation of first-order programs
using algebraic data types and pattern matching
is given in Section~\ref{sec:code}.
In these sections, we assume that data types are finite
(e.g., composed from \verb|Bool|, \verb|Maybe|, \verb|Either|),
and programs are total. 
We then extend this 
in section \ref{sec:code-recursive} 
to handle infinite (that is, recursive) data types
(e.g., lists, trees).

We then treat briefly two ideas that serve to improve writing
and executing CO4 programs:
In Section~\ref{sec:remove}, we discuss the compilation 
of higher-order and polymorphic features in CO4 programs.
In Section~\ref{sec:memo}, we show that memoization of function calls
improves efficiency since it allows to share sub-formulas.

With these preparations, we give the CO4 formulation of looping derivations
in term %
rewriting systems in Section~\ref{sec:loop}.
Propositional encodings for string rewrite sequences have appeared
in the literature~\cite{DBLP:conf/sofsem/ZanklSHM10}.
To our knowledge, the propositional encoding of term rewriting is new,
and it looks quite an insurmountable task to write such an encoding 
without the help of a compilation system.

\newcommand{\PP}{\mathbb{P}}
\newcommand{\QQ}{\mathbb{Q}}
\newcommand{\CC}{\mathbb{C}}
\renewcommand{\AA}{\mathbb{A}}
\newcommand{\BB}{\mathbb{B}}
\newcommand{\cvalue}{\concreteValue} %
\newcommand{\avalue}{\abstractValue} %

\section{Semantics of Propositional Encodings}
\label{sec:semantics}

In this section give the specification for compilation of CO4 expressions,
in the form of an invariant (it should hold for all sub-expressions).
When applied to the full input program, the specification implies
that the compiler works as expected: 
a solution for the constraint system can be found via the external SAT solver.
We defer discussion of our implementation of this specification 
to Section~\ref{sec:code}, and give here a more formal, but still
high-level view of the CO4 language and compiler.

\paragraph{Evaluations on concrete data.}
We denote by $\PP$ the set of expressions in the input language.
It is a first-order functional language with
\begin{itemize}
\item algebraic data types,
\item pattern matching,
\item global and local function definitions (using \verb|let|)
  that may be recursive.
\end{itemize}
The concrete syntax is a subset of Haskell. We give examples---
which may appear unrealistically simple
but at this point we cannot use higher-order or polymorphic features.
These will be discussed in see Section~\ref{sec:remove}.

\begin{verbatim}
data Bool = False | True
and2 :: Bool -> Bool -> Bool
and2 x y = case x of { False -> False ; True -> y }

data Maybe_Bool = Nothing | Just Bool

f :: Maybe_Bool -> Maybe_Bool -> Maybe_Bool
f p q = case p of
    Nothing -> Nothing
    Just x -> case q of
        Nothing -> Nothing
        Just y -> Just (and2 x y)
\end{verbatim}

For instance, \verb|f (Just x) Nothing| is an expression of $\PP$,
containing a variable \verb|x|. 
We allow only \emph{simple} patterns (a constructor followed by variables), 
and we require that pattern matches are \emph{complete}
(there is exactly one pattern for each constructor of the respective type).
It is obvious that nested patterns can be translated to this form.

Evaluation of expressions is defined in the standard way:
The domain of \emph{concrete values} $\CC$ is the set of data terms.
For instance, \texttt{Just False}$\ \in \CC$.
A \emph{concrete environment} 
is a mapping from program variables to $\CC$.
A \emph{concrete evaluation function} $\cvalue:E_\CC \times \PP \to \CC$
computes the value of a concrete expression $p\in\PP$
in a concrete environment $e_\CC$.
Evaluation of function and constructor arguments is strict.
This is where we deviate from Haskell's lazy evaluation.

\paragraph{Evaluations on abstract data.}

The CO4 compiler transforms an input program that operates on concrete values,
to an \emph{abstract program} that operates on \emph{abstract values}.
An abstract value contains propositional logic formulas 
that may contain free propositional variables.
An abstract value represents a set of concrete values. 
Each assignment of the propositional values produces a concrete value.

We formalize this in the following way:
the domain of abstract values is called $\AA$.
The set of assignments 
(mappings from propositional variables to truth values $\BB=\{0,1\}$)
is called $\Sigma$,
and there is a function $\decode:\AA\times \Sigma \to\CC$.

We now specify abstract evaluation. 
(The implementation is given in Section~\ref{sec:code}.)
We use \emph{abstract environments} $E_\AA$
that map program variables to abstract values, 
and an \emph{abstract evaluation function}
$\avalue: E_\AA \times \PP \to \AA$.

\paragraph{Allocators.} As explained in the introduction,
the constraint program receives known and unknown arguments.
The compiled program operates on abstract values.

The abstract value that represents a (finite) set of concrete values
of an unknown argument is obtained from an \emph{allocator}.
For a property $q : \CC\to\BB$ of concrete values,
a  $q$-allocator constructs an object $a\in\AA$
that represents all concrete objects that satisfy $q$:
\[ \forall c\in \CC: q(c) \iff \exists \sigma\in\Sigma : c=\decode(a,\sigma). \]

We use allocators for properties $q$ that specify 
$c$ uses constructors that belong to a specific type.
Later (with recursive types, see Section~\ref{sec:code-recursive})
we also specify a size bound for $c$.
An example is an allocator for lists of booleans of length $\le 4$.

As a special case, an allocator for a singleton set
is used for encoding a known concrete value.
This \emph{constant allocator} is given by a function $\encode:\CC\to\AA$ 
with the property that $\forall c\in\CC, \sigma\in\Sigma: \decode(\encode(c),\sigma)=c$.

\paragraph{Correctness of constraint compilation.}

The semantical relation between an expression $p$ (a concrete program)
and its compiled version $\compile(p)$ (an abstract program) is given by
the following relation between concrete and abstract evaluation:

\begin{definition} \label{def:invariant}
We say that $p\in\PP$ is compiled \emph{correctly}
  if 
\begin{equation}
  \begin{aligned}
  \label{eq:invariant}
  \forall e\in E_\AA \ \forall \sigma \in\Sigma: &
              \decode (\avalue (e, \compile(p)),\sigma)\\
            = &\cvalue (\decode(e,\sigma), p)
  \end{aligned}
\end{equation}
\end{definition}

Here we used $\decode(e,\sigma)$ as notation for lifting the
decoding function to environments, defined element-wise by 
\begin{equation*}
  \label{eq:lift-decode}
  \forall e\in E_\AA \ \forall v \in \operatorname{dom}(e) \ \forall \sigma \in\Sigma: 
  \decode(e,\sigma)(v) = \decode(e(v),\sigma).
\end{equation*}

\paragraph{Application of the Correctness Property.}

We are now in a position to show how the stages of CO4 compilation 
and execution fit together.

The top-level parametric constraint is given by a function declaration
\verb|main k u = b|
where \verb|b| (the \emph{body}, a concrete program) is of type \verb|Bool|.
It will be processed in the following stages:
\begin{enumerate}
\item \emph{compilation} produces an abstract program $\compile(b)$,
\item \emph{abstract computation} takes a concrete parameter value $p\in \CC$
  and a $q$-allocator $a\in \AA$, and computes the formula
\[ F =\avalue( \{ k \mapsto \encode (p), u \mapsto a \} , \compile(b)) \]
\item \emph{solving} calls the backend SAT solver
  to determine $\sigma\in\Sigma$ with  $\decode(F,\sigma)=\true$. 
  If this was successful,
\item \emph{decoding} produces a concrete value $s = \decode(a,\sigma)$,
\item and optionally, \emph{testing} checks that 
  $\cvalue (\{ k \mapsto p, u \mapsto s\}, b)= \true$.
\end{enumerate}
The last step is just for reassurance against implementation errors,
since the invariant implies that the test returns True.
This highlights another advantage of re-using Haskell for constraint programming:
one can easily check the correctness of a solution candidate.

\section{Implementation of a Propositional Encoding} \label{sec:code}

In this section, we give a realiziation for abstract values,
and show how compilation creates programs that operate correctly
on those values, as specified in Definition~\ref{def:invariant}

\paragraph{Encoding and Decoding of Abstract Values.}

The central idea is to represent an abstract value as a tree,
where each node contains an encoding for a symbol (a constructor)
at the corresponding position,
and the list of concrete children of the node
is a prefix of the list of abstract children
(the length of the prefix is the arity of the constructor).

The encoding of constructors is by a sequence of formulas
that represent the number of the constructor in binary notation.

We denote by $\set{F}$ the set of propositional logic formulas.
At this point, we do not prescribe a concrete representation.
For efficiency reasons, we will allow some form of sharing,
by representing formulas as directed acyclic graphs
(e.g., and/inverter graphs). 
Our implementation (\verb|satchmo-core|) 
assigns names to subformulas by doing the Tseitin transform on-the-fly,
creating a fresh propositional literal for each subformula.

\begin{definition}
  \label{def:abstract-value}
  The set of abstract values $\AA$ is the smallest set with  $\AA = \set{F}^* \times \AA^*$.

  An element $a\in \AA$ thus has shape $(\arrow{f},\arrow{a})$
  where $\arrow{f}$ is a sequence of formulas, called the \emph{flags} of $a$,
  and $\arrow{a}$ is a sequence of abstract values, called the \emph{arguments} of $a$.

  We introduce notation
  \begin{itemize}\itemsep 0pt \parskip 2pt
    \item $\flags : \mathbb{A} \to \set{F}^*$ gives the flags of an abstract value
    \item $\flags_i : \mathbb{A} \to \set{F}$ gives the $i$-th flag of an abstract value
    \item $\arguments : \AA\to\AA^* $ gives the arguments of an abstract value,
    \item $\argument_{i} : \mathbb{A} \to \mathbb{A}$ gives the
      $i$-th argument of an abstract value
  \end{itemize}

  Equivalently, in Haskell notation, 
\begin{verbatim}
data A = A { flags :: [F] , arguments :: [A] }
\end{verbatim}
\end{definition}

The sequence of flags of an abstract value encodes the number of its constructor.
We use the following variant of a binary encoding:
For each data type $T$ with $c$ constructors, 
we use as flags a set of sequences $S \subseteq \{0,1\}^*$ with $|S|=c$
and such that each long enough $w\in\{0,1\}^*$ does have exactly one prefix in $S$.

We could have $S$ dependent on $T$, but this is not necessary.
In practice we use a fixed encoding
\begin{equation*}
  \begin{aligned}
    S_1 = \{\epsilon\}; \qquad
    \text{for $n>1$:} \quad S_n = 0\cdot S_{\lceil n/2 \rceil} \cup 1 \cdot S_{\lfloor n/2\rfloor} 
  \end{aligned}
\end{equation*}
For example, 
$S_2=\{0,1\}, 
S_3 = \{00, 01,1\},
S_5 = \{000,001,01,10,11\}$.
The lexicographic order of $S_c$ induces a bijection 
$\numeric_c:S_c\to\{1,\ldots,c\}$.

The encoding function (from concrete to abstract values)
is defined by 
\begin{equation*}
  \encode_T(C(v_1,\ldots)) = (\numeric_c^-(i), [\encode_{T_1}(v_1),\ldots])
\end{equation*}
where $C$ is the $i$-th constructor of type $T$,
and $T_j$ is the type of the $j$-th argument of $C$.
Note that here, $\numeric_c^-(i)$ denotes a sequence of
constant flags (formulas) that represents the corresponding binary string.

For decoding, we need to take care of extra flags and arguments
that may have been created
by the function $\merge$ (Definition~\ref{def:merge})
that is used in the compilation of \verb|case| expressions.

We extend the mapping $\numeric_c$ to longer strings by
$\numeric_c(u\cdot v) := \numeric_c(u)$ for each $u\in S_c,v\in\{0,1\}^*$.
This is possible because of the unique-prefix condition.

Given the type declarations
\begin{verbatim}
data Bool = False | True
data Maybe_Bool = Nothing | Just Bool
data Ordering = LT | EQ | GT
data Either_Bool_Ordering = Left Bool | Right Ordering
\end{verbatim}
the concrete value \verb|True| can be represented
by the abstract value $a_1=([x],[])$ and assignment $\{x=1\}$,
since \verb|True| is the second (of two) constructors,
and $\numeric_2([1])=2$.
The same concrete value \verb|True| can also be represented
by the abstract value $a_2=([x,y],[a_1])$ and assignment $\{x=1,y=0\}$,
since $\numeric_2([1,0])=2$. This shows that extra flags and extra
arguments are ignored in decoding.

We give a formal definition:
for a type $T$ with $c$ constructors,
$\decode_T(({f},{a}),\sigma)$ 
is the concrete value $v= C_i (v_1,\ldots)$ 
where $i=\numeric_c({f}\sigma)$,
and $C_i$ is the $i$-th constructor of $T$,
and $v_j = \decode_{T_j}(a_j,\sigma)$
where $T_j$ is the type of the $j$-th argument of $C_i$.

As stated, this is a partial function, 
since any of ${f}, {a}$ may be too short.
For this Section, we assume that abstract values always have
enough flags and arguments for decoding, 
and we defer a discussion of partial decodings to Section~\ref{sec:code-recursive}.

\paragraph{Allocators for Abstract Values.}

Since we consider (in this section) finite types only,
we restrict to \emph{complete} allocators:
for a type $T$, a complete allocator is an abstract value $a\in \AA$
that can represent each element of $T$:
for each $e\in T$, there is some $\sigma$ such that $\decode_T(a,\sigma)=e$.

For the types given above, complete allocators are

\begin{center}
  \begin{tabular}{l|l}
    type & complete allocator \\ \hline 
    $\code{Bool}$ & $a_1=([x_1], [])$ \\
    $\code{Ordering}$ & $a_2=([x_1,x_2], [])$ \\
    $\code{Either\_Bool\_Ordering}$ & $a_3=([x_1], [([x_2,x_3],[])])$ \\
  \end{tabular}
\end{center}

where $x_1,\ldots$ are (boolean) variables.
We compute $\decode(a_3,\sigma)$ for $\sigma=\{x_1=0,x_2=1,x_3=0\}) $:
Since $\numeric_2([0])=1$, the top constuctor is \verb|Left|.
It has one argument, obtained as $\decode_{\texttt{Bool}}(([x_2,x_3],[]),\sigma)$.
For this we compute $\numeric_2([1,0])=2$,
denoting the second constructor (\verb|True|) of \verb|Bool|.
Thus, $\decode(a_3,\sigma)=\texttt{Left True}$.

\paragraph{Compilation of Programs.}

In the following we illustrate the actual transformation of the input program
(that operates on concrete values)
to an abstract program (operating on abstract values)
and prove its soundness according to invariant (Definition~\ref{def:invariant}).

Generally, compilation keeps structure and names of the program intact.
For instance, if the original program defines functions $f$ and $g$,
and the implementation of $g$ calls  $f$,
then the transformed program also defines functions $f$ and $g$,
and the implementation of $g$ calls $f$.

The crucial exception is that \emph{compilation removes pattern matches}.
This is motivated as follows.
Concrete evaluation of a pattern match (in the input program)
consists of choosing a branch according to a concrete value 
(of the discriminant expression).
Abstract evaluation cannot access this concrete value
(since it will only be available after the SAT solver determines an assignment).
This means that we cannot abstractly evaluate pattern matches.
Therefore, they must be removed by compilation.

Compilation of variables, bindings, and function calls
is straightforward, and we deal with them first.

\begin{definition}[Compilation, easy cases]
  \begin{itemize}
  \item a name is compiled into itself: 

    if $v$ is a variable, then $\compile(v)=v$.
  \item a local binding is compiled structurally:

    $\compile(\code{let}~v = a~\code{in}~b) 
    ~= ~
    \code{let}~v =  \compile(a)~\code{in}~\compile(b)$

  \item a function call is compiled structurally:

    $\compile (f(a_1,\ldots,a_n)) = f(\compile(a_1),\ldots,\compile(a_n))$

    Here, compilation creates an application of $f$.
    It is executed during abstract evaluation.
  \end{itemize}
\end{definition}
\begin{lemma}[Correctness of compilation, easy cases]
  Invariant \eqref{eq:invariant} holds on compilation of variables, local bindings
  and function calls.
\end{lemma}

\begin{proof}
  Let $v$ be a variable of the input program and $\compile(v) = v$.
  As the abstract value of $v$ only depends on the value of $v$, i.e.
  \begin{equation*}
    \forall e \in \absEnv: \abstractValue(e, v) = e (v),
  \end{equation*}
  and the concrete value of the original expression $v$ only depends on the
  value of $v$ as well, i.e.

  \begin{equation*}
    \forall e \in \absEnv: \concreteValue (\decode (e,\sigma), v) = \decode (e,\sigma) (v),
  \end{equation*}
  by \eqref{eq:lift-decode} we have
  \begin{align*}
    \forall e \in \absEnv \forall \sigma \in \Sigma: 
         &\decode_T (\abstractValue(e,\compile(v)),\sigma) \\
      =\ &\decode_T (\abstractValue(e,v),\sigma) \\
      =\ &\decode_T (e (v),\sigma) \\
      =\ &\decode (e,\sigma) (v) \\
      =\ &\concreteValue (\decode (e,\sigma), v).
  \end{align*}
  So invariant \eqref{eq:invariant} holds.

The proof of correctness of compilation of
local bindings and function calls is by structural induction.
  \qed
\end{proof}

\begin{definition}[Compilation, constructor call]

For a constructor call $C(p_1,\ldots,p_n)$
where $C$  is the $i$-th constructor of a data type $T$ 
(with $c$ constructors in total)
and $p_j$ is of type $T_j$,

\begin{equation*}
\compile(C(p_1,\ldots,p_n)) = C'(\compile(p_1),\ldots,\compile(p_n))
\end{equation*}

where $C' : \AA^* \to \AA$ is a function that gets the abstract values 
$(a_1,\ldots,a_n)$ of the compiled constructor arguments as input.

\begin{equation}
  \label{eq:cons-abstract-value}
  C' (a_1,\ldots,a_n) = (\numeric_c^-(i), [a_1,\ldots,a_n])
\end{equation}

Note that $C'$ is evaluated during the runtime of the abstract program.

\end{definition}

\begin{lemma}[Correctness of compilation, constructor call]
  Invariant \eqref{eq:invariant} holds on compilation of constructor calls.
\end{lemma}

\begin{proof}
  If $C$ is the $i$-th constructor of a type $T$, the decoding of 
  $\compile(C(p_1,\dots,p_n))$'s top level constructor is determined
  by the fixed flags of its abstract value \eqref{eq:cons-abstract-value}.

  \begin{equation}
    \label {eq:proof-con-abstract-value}
    \begin{aligned}
      \forall e \in \absEnv &\forall \sigma \in \Sigma: \\
        & \decode_T (\abstractValue (e, \compile (C (p_1,\dots,p_n))),\sigma) \\
      =\ & \decode_T (C' (\compile (p_1),\dots,\compile(p_n)),\sigma) \\
      =\ & C (\decode_T (\compile (p_1)),\dots,\decode_T (\compile (p_1)))\\
    \end{aligned}
  \end{equation}

  The top-level constructor of $C(p_1,\dots,p_n)'s$ concrete value is independent
  of any environment, so

  \begin{equation}
    \label {eq:proof-con-concrete-value}
    \begin{gathered}
      \forall e \in \absEnv \forall \sigma \in \Sigma: 
      \concreteValue (\decode (e,\sigma), C(p_1,\dots,p_n)) = \\
      C (\concreteValue (\decode (e,\sigma),p_1),\dots,\concreteValue (\decode (e,\sigma),p_n))
    \end{gathered}
  \end{equation}.

  The equality of \eqref{eq:proof-con-abstract-value} and 
  \eqref{eq:proof-con-concrete-value} is proven by induction over the constructor
  arguments.
  \qed
\end{proof}

We restrict to pattern matches where patterns are 
\emph{simple} (a constructor followed by variables)
and \emph{complete} (one branch for each constructor of the type).

\begin{definition}[Compilation, pattern match]

Consider a pattern match expression $e$ of shape
$\code{case}~d~\code{of}~\{ \dots \}$,
for a discriminant expression $d$ of type $T$ with $c$ constructors.

We have 
$\compile(e)=\code{let}~x = \compile(d)~\code{in}~\merge_c(\flags(x),b_1,\ldots)$
where $x$ is a fresh variable, 
and $b_i$ represents the compilation of the $i$-th branch.

Each such branch is of shape $C\ v_1 \ldots v_n \to e_i$,
where $C$ is the $i$-th constructor of the type $T$.

Then $b_i$ is obtained as
$\code{let}~\{ v_1 = \argument_1(x); \dots \}~\code{in}~\compile(e_i)$.

\end{definition}

We need the following auxiliary function
that combines the abstract values from branches of pattern matches,
according to the flags of the discriminant.

\begin{definition}[Combining function]\label{def:merge}
  $\merge : F^*\times \AA^c\to \AA$ combines abstract values so that
  $\merge(\overrightarrow{f},a_1,\ldots,a_c)$ 
  is an abstract value $(\overrightarrow{g},z_1,\dots,z_n)$, where

  \begin{itemize}
    \item $n = \max (|\arguments(a_1)|,\dots,|\arguments (a_c)|)$
    \item $|\arrow{g}| = \max (|\flags(a_1)|,\dots,|\flags (a_c)|)$
    \item for $1\le i\le |\arrow{g}|$,
      \begin{equation}      \label{eq:merge-flag}
      \begin{aligned}
        g_i \leftrightarrow \quad 
                  & (\numeric_c (\overrightarrow{f}) = 1 \Rightarrow \flags_i (a_1))\\
          \land \ & (\numeric_c (\overrightarrow{f}) = 2 \Rightarrow \flags_i (a_2))\\
          \land \ & \dots \\
          \land \ & (\numeric_c (\overrightarrow{f}) = c \Rightarrow \flags_i (a_c))\\
      \end{aligned}
      \end{equation}
    \item for each $1\le i\le n$, 
      $z_i = \merge(\overrightarrow{f},\argument_i (a_1),\ldots,\argument_i (a_c))$.

  \end{itemize}
\end{definition}

\begin{lemma}[Correctness of compilation, pattern match]
  Invariant \eqref{eq:invariant} holds on compilation of pattern matches.
\end{lemma}

\begin{proof}
  Let $m$ be a pattern match in the original program and $m'$ the result of
  the corresponding $\merge$.
  The decoding of $m'$ depends on an assignment $\sigma$.
  For any assignment $\sigma$ there is a $k \in [1,c]$ so that 
  $\numeric_c (\flags (m')) = k$.
  In this case, \eqref{eq:merge-flag} shows that for all flags of $m'$ 
  $\flags_i (m') \leftrightarrow \flags_i (a_k)$ holds, with $b_k$ being the
  $k$-th branch of $e$ and $a_k = \compile (b_k)$.  
  So, for a fixed $\sigma$ (and hence $k$) property

  \begin{equation*}
    \forall e \in \absEnv : 
      \decode (m',\sigma) = \decode (\abstractValue (e,a_k),\sigma)
  \end{equation*} 
  holds.

  As evaluating the concrete value of the original pattern match $m$ under
  an environment decoded by $\sigma$ leads to the evaluation of $b_k$, i.e.

  \begin{equation*}
    \forall e \in \absEnv : 
      \concreteValue (\decode (e,\sigma), m) = \concreteValue (\decode (e,\sigma),b_k) \\
  \end{equation*}
  .
  Invariant \eqref{eq:invariant} holds by induction over $b_k$:
  \begin{gather*}
    \forall e \in \absEnv : 
      \decode (\abstractValue (e,a_k),\sigma) = \concreteValue (\decode (e,\sigma),b_k) \\
  \end{gather*}
\end{proof}

\section {Partial encoding of Infinite Types}\label{sec:code-recursive}

We discuss the compilation and abstract evaluatation for constraints over infinite types,
like lists and trees. Consider declarations
(and recall that functions are still monomorphic)
\begin{verbatim}
data N = Z | S N
double :: N -> N
double x = case x of { Z -> Z ; S x' -> S (S (double x')) }
\end{verbatim}
Assume we have an abstract value $a$ to represent \verb|x|.
It consists of a flag (to distinguish between \verb|Z| and \verb|S|),
and of one child (the argument for \verb|S|), which is another abstract value.
At some depth, recursion must stop, since the abstract value is finite
(it can only contain a finite number of flags).
Therefore, there is a child with no arguments,
and it must have its flag set to $[\false]$ (it must represent \verb|Z|).

There is another option: if we leave the flag open (it can take on
values $\false$ or $\true$), then we have an abstract value with (possibly) 
a constructor argument missing. When evaluating the concrete program,
the result of accessing a non-existing component gives a bottom value.
This corresponds to the Haskell semantics 
where each data type contains bottom,
and values like $\code{S}~(\code{S}~ \bot)$ are valid.

\newcommand{\AAbot}{\AA_{\bot}}
\newcommand{\PPbot}{\PP_{\bot}}
\newcommand{\QQbot}{\QQ_{\bot}}

\begin{definition}
  The set of abstract values $\AAbot$ is the smallest set with  
  $\AAbot = \set{F}^* \times \AAbot^* \times \set{F}$, i.e. an abstract value
  is a triple of flags and arguments (cf. definition \ref{def:abstract-value}) 
  extended by an additional \emph{definedness constraint}.

  We write $\defness : \AAbot \to \set{F}$ to give the definedness constraint
  of an abstract value, and keep 
  $\flags$ and $\argument$ notation of  Definition~\ref{def:abstract-value}.
\end{definition}

The decoding function is modified accordingly: 
$\decode_T(a,\sigma)$
for a type $T$ with $c$ constructors
is $\bot$ if $\defness(a)\sigma=\false$, 
or $\numeric_c(\flags(a))$ is undefined (because of ``missing'' flags), 
or  $|\arguments(a)|$ is less than the number of arguments
of the decoded constructor.

The correctness invariant for compilation (Eq. \ref{eq:invariant}) is still the same,
but we now interpret it in the domain $\CC_\bot$,
so the equality says that if one side is $\bot$, then both must be.

Consequently, for the application of the invariant,
we now require that the abstract value of the top-level constraint
under the assignment is defined and $\true$.

\paragraph{Abstract evaluation with bottoms.}
For working with recursive types, we need recursive programs.
If the input program is recursive, then so is the abstract program.
Pattern matching is crucial to terminate recursion
(e.g., to detect the end of a list), 
but the abstract program cannot pattern match, as explained earlier.

We introduce a limited form of matching: 
in the abstract evaluation of $\code{let} x = \compile(d) \code{in} \dots$ 
(see Compilation, pattern match), we consider $\code{let}$ to be strict:
if $\abstractValue(E,\compile(d))$ 
has a definedness flag that is constant $\false$, 
then the whole expression's abstract value is bottom,
and is represented by $([],[],\false)$.

If definedness is not constantly $\false$, 
then abstract evaluation will execute $\merge$,
modified as follows: the definedness flag of result $m$ of a $\merge$ is
      \begin{equation*} 
      \begin{aligned}
        \defness(m)
         \leftrightarrow \quad 
                  & (\numeric_c (\overrightarrow{f}) = 1 \Rightarrow \defness (a_1))\\
          \land \ & (\numeric_c (\overrightarrow{f}) = 2 \Rightarrow \defness (a_2))\\
          \land \ & \dots \\
          \land \ & (\numeric_c (\overrightarrow{f}) = c \Rightarrow \defness (a_c))\\
      \end{aligned}
      \end{equation*}

Note that (definedness and other) flags are formulas,
and in general we cannot determine their value (without an assignment).
The given method relies on some form to detect that a formula denotes a constant.

\section{Higher order functions and polymorphism}\label{sec:remove}

For formulating the constraints, expressiveness in the language is welcome.
Since we base our design on Haskell, it is natural to include some of its
features that go beyond first-order programs: higher order functions
and polymorphic types. 

Our program semantics is first-order:
we cannot (easily) include functions as result values or in environments,
since we have no corresponding abstract values for functions.
Therefore, we instantiate all higher-order functions
in a standard preprocessing step, starting from the main program.

Polymorphic types do not change the compilation process.
The important information is the same as with monomorphic typing:
the total number of constructors of a type, 
and the number (the encoding) of one constructor.

\section{Memoization}\label{sec:memo}

We describe another optimization: in the abstract program,
we use memoization for all subprograms.
That is, during execution of the abstract program,
we keep a map from (function name, argument tuple) to result.
Note that arguments and result are abstract values.

This allows to write ``natural'' specifications
and still get a reasonable implementation.

\begin{example}
  The textbook definition of the lexicographic path order $>_{\textit{lpo}}$ (cf. \cite{Baader})
  defines an order over terms according to some precedence.
  Its textbook definition is recursive, and leads to an
  exponential time algorithm, if implemented literally.
  By calling $s >_{\textit{lpo}} t$ the algorithm still does only compare 
  subterms of $s$ and $t$,
  and in total, there are $|s|\cdot |t|$ pairs of subterms,
  and this is also the cost of the textbook algorithm with memoization.
\end{example}

The next example is similar, but it additionally shows 
that abstract execution may increase cost,
but memoization may reduce it again.

\begin{example}
The following function determines
whether \verb|xs| is a (scattered) subword of \verb|ys|.

\begin{verbatim}
subword :: Eq a => [a] -> [a] -> Bool
subword xs ys = case xs of
    [] -> True
    x : xs' -> case ys of
        [] -> False
        y : ys' -> case x == y of
            False -> subword  xs ys'
            True  -> subword xs' ys'
\end{verbatim}

As a program on concrete values, this has linear complexity,
since in each recursive call, the length of the second argument decreases.

In the compiled program for abstract values,
for each \verb|case|, each branch is executed, and the results are merged.
In particular,
both branches of \verb|case x==y of ..| will be executed,
so the resulting cost is exponential in the size of \verb|ys|.

With memoization, the compiled program runs in polynomial time
(and produces a polynomially sized formula)
since in each subprogram call that happens 
during the evaluation of \verb|subword xs0 ys0|,
the actual arguments \verb|xs, ys| 
are suffixes of the respective initial arguments,
and there are (\verb|length xs * length ys|) pairs of suffixes.

\end{example}

\section{Case study: Loops in Term Rewriting}\label{sec:loop}

As an application, we use CO4 for compiling
constraint systems that describe looping derivations.
This is motivated by automated analysis of programs.
A loop is an infinite computation,
which may be unwanted behaviour,
indicating an error in the program's design.
In general, it is undecidable whether 
a rewriting system admits a loop.
By enumerating finite derivations, 
one can hope to find loops.

Our approach is to write the predicate
``the derivation $d$ conforms to a rewrite system $R$
and $d$ is looping'' as a Haskell function,
and solve the resulting constraint system,
after putting bounds on the sizes of the terms
that are involved.

Previous work uses several heuristics for enumerations resp.
hand-written propositional encodings for finding loops in 
string rewriting systems~\cite{DBLP:conf/sofsem/ZanklSHM10}.

We extend to (1) systematic compilation
and (2) term rewriting.

In the following, we show the data declarations we use,
and give code examples.
\begin{itemize}
\item we fix a signature, and a set of variables, and define the set of terms
\begin{verbatim}
data Term = V Name | F Term Term Term | A | B | C
data Name = X | Y
\end{verbatim}
\item a rule is pair of terms, a rewrite system is list of rules
\begin{verbatim}
  data Rule = Rule Term Term
  data List a = Nil | Cons a (List a)
  type TRS = List Rule
\end{verbatim}
\item 
    a rewrite step is a tuple $(t_0, (l,r), P, \sigma, t_1)$
    where $t_0,t_1$ are terms, $(l,r)$ is a rule, $p$ is a position, 
    $\sigma$ is a substitition
    with $l\sigma = t_0[p]$ and $t_0[p := r\sigma] = t_1$

    \begin{verbatim}
  data Pair a b = Pair a b
  type Substitution = List (Pair Name Term)
  data Step = Step Term Rule (List Pos) Substitution Term
    \end{verbatim}

\item a derivation w.r.t. a TRS \texttt{trs} is a list of steps
  \begin{verbatim}
  type Derivation = List Step
  \end{verbatim}
  where
  \begin{itemize}
  \item the result \texttt{term} of one step is the input term of the next step 
  \begin{verbatim}
  derive_ok :: TRS -> Term -> Derivation -> Maybe Term
  derive_ok trs term deriv = case deriv of
    Nil           -> Just term
    Cons s deriv' -> case s of
      Step t0 rule pos sub t1 -> case equalTerm term t0 of
        False -> Nothing
        True  -> case step_ok trs s of 
                    False -> Nothing
                    True  -> derive_ok trs t1 deriv'
  \end{verbatim}

  \item each step's rule is from the \texttt{trs}
  \end{itemize}

\item a looping derivation's output of the last step
    has a subterm that is a substitution instance of the input of the first step
\end{itemize}

Overall, the complete CO4 code 
(available at
\url{https://github.com/apunktbau/co4/blob/master/CO4/Test/TRS_Loop_Toyama.standalone.hs})
consists of roughly 300 lines of code 
including the definition of
all involved data types and auxiliary functions. 
The code snippets above shows that the constraint system
literally follows the textbook definitions.

Our test case is the following term rewriting system,
where $X,Y$ are variables,
\begin{equation*}
  \{
  f (a,b,X) \to f (X,X,X), 
  f (X,Y,c)   \to X ,
  f (X,Y,c)   \to Y \},
\end{equation*}
(corresponding to the classical example from
\cite{DBLP:journals/ipl/Toyama87}).
We use allocators that restrict to derivations of length 3, and terms of depth 2.
Abstract evaluation of the compiled program
results in a propositional formula with 774663 variables and 2301608 clauses.
On a standard Intel Core 2 Duo CPU with 2.20\,GHz, Minisat SAT solver
finds the following loop in around 10 seconds:

\begin{align*}
  &f (a,b,f (a,b,c)) \Rightarrow 
  f (f (a,b,c), f(a,b,c),f(a,b,c)) \\ \Rightarrow
  &f (a,f (a,b,c),f(a,b,c)) \Rightarrow 
  f (a,b,f(a,b,c))
\end{align*}

\footnotesize
\begin{verbatim}
CNF finished (#variables: 774663, #clauses: 2301608)
Solver finished in 10.823333 seconds (result: True)
Solution: Looping_Derivation ...
Test: True
\end{verbatim}
\normalsize

\section{Discussion}

In this paper we described the CO4 constraint language and compiler
that allows to write constraints on tree-shaped data in a natural way, 
and to solve them via propositional encoding.
We presented an outline of a correctness proof for our implementation,
and gave an example that shows that the compiler actually works.

In this example, the resulting formula is huge.
Still, the SAT solver can handle it rather quickly.
This indicates that there is room for improving the efficiency
of both compilation and abstract interpretation,
in order to obtain smaller, equivalent, formulas
from constraint systems.

We have several plans for this, for instance, 
hard-wiring improved implementations of basic boolean operations.
We leave this as a subject of further research and implementation,
for which the present report shall serve as a basis.

We mention two additional application areas of the concepts presented here:
\begin{itemize}
\item Different back-ends:
  In our example application, formulas (in abstract values) are ultimately
  represented as conjunctive normal forms, as this suits the SAT solver best.
  By changing the implementation of abstract values (but keeping the compiler),
  our system can output circuit descriptions for hardware design;
  and also Binary Decision Diagrams, which can be used for counting models
  of constraint systems.
\item Complexity analysis: from an (automated) analysis of the input program,
  one can obtain the (asymptotic) size  of the resulting propositional formula.
  For instance,
  if it is found to be polynomial (in the size of the input parameter),
  then satisfiability of the constraint problem is (automatically shown to be) in NP.
\end{itemize}

\bibliographystyle{plain}
\bibliography{bib}

\end{document}